\documentclass[11pt,a4paper,leqno]{article}%
\usepackage[left=2cm, top=2cm, bottom=2cm, right=2cm]{geometry}
\usepackage{amsmath,amssymb,amscd,amsthm}
\usepackage[colorlinks=true ,urlcolor=blue ,anchorcolor=blue ,citecolor=blue
,filecolor=blue ,linkcolor=blue ,menucolor=blue ,pagecolor=blue
,linktocpage=true ,pdfproducer=medialab ]{hyperref}
\usepackage{enumitem}%
\usepackage{amsmath}%
\usepackage{amsfonts}%
\usepackage{amssymb}%
\usepackage{graphicx}
\providecommand{\U}[1]{\protect\rule{.1in}{.1in}}
\numberwithin{equation}{section}
\theoremstyle{plain}
\newtheorem{theorem}{Theorem}[section]
\newtheorem{lemma}[theorem]{Lemma}

\newtheorem{proposition}[theorem]{Proposition}

\theoremstyle{definition}

\newenvironment{proof**}[1]{\paragraph{\it Proof of #1.}}{\hfill $\square$\\}
\newcommand{\la}{\lambda}

\newcommand{\map}{\mapsto}

\newcommand{\pr}{\partial}
\newcommand{\me}{\geqslant}

\newcommand{\bra}[1]{\! \left (#1\right )}

\DeclareMathOperator{\diag}{diag}

\newcommand{\eqq}[1]{\begin{align*} #1 \end{align*}}
\newcommand{\eq}[1]{\begin{align} #1 \end{align}}
\setlist{nolistsep}
\begin{document}

\title{\vspace{0mm}\textsc{Integrable quantum St\"ackel systems}}
\author{Maciej B\l aszak$^{\dag}$, Ziemowit Doma\'nski$^{\dag}$, Artur
Sergyeyev$^{\ddag}$, B\l a\.zej M. Szablikowski$^{\dag}$\\{\small \dag Faculty of Physics, Adam Mickiewicz University}\\{\small Umultowska 85, 61-614 Pozna\'n, Poland}\\{\small e-mails: \texttt{blaszakm@amu.edu.pl}, \texttt{ziemowit@amu.edu.pl},
\texttt{bszablik@amu.edu.pl}}\\[3mm] {\small \ddag Mathematical Institute, Silesian University in Opava}\\{\small Na Rybn\'\i{}\v{c}ku 1, 746\,01 Opava, Czech Republic}\\{\small e-mail: \texttt{Artur.Sergyeyev@math.slu.cz}}}
\date{}
\maketitle

\begin{abstract}
The St\"ackel separability of a Hamiltonian system is well known to ensure
existence of a complete set of Poisson commuting integrals of motion quadratic
in the momenta. We consider a class of St\"ackel separable systems where the
entries of the St\"ackel matrix are monomials in the separation variables. We
show that the only systems in this class for which the integrals of motion
arising from the St\"ackel construction keep commuting after quantization are,
up to natural equivalence transformations, the so-called Benenti systems.
Moreover, it turns out that the latter are the only quantum separable systems
in the class under study.\looseness=-1

\end{abstract}

{\small Corresponding author: Maciej B\l aszak, tel. +48 618295057, e-mail: \texttt{blaszakm@amu.edu.pl}}

\section{Introduction}

The separation of variables is well known to be one of the most powerful
methods for integration of equations of motion for dynamical systems, see
e.g.\ \cite{M, MK, Sklyanin, B} and references therein. In classical mechanics
the separation of variables occurs in the Hamilton--Jacobi equation and in
quantum mechanics in the Schr\"odinger equation, with the former being a
classical limit of the latter. The necessary and sufficient conditions for
orthogonal separation of variables for natural Hamiltonians in absence of
magnetic field in classical mechanics were established by St\"ackel \cite{St}.\looseness=-1

While the separation of variables in the Schr\"odinger equation implies one in
the Hamilton--Jacobi equation, the converse is not true. This is, however, a
fairly subtle issue because there is a number of not entirely equivalent
definitions of separation of variables, see the discussion in \cite{Ben3,
Ben4} and also e.g.\ \cite{shap, bg, ZZ} and references therein. In the
present paper we stick to the definition used in \cite{Ben3, Ben4}. With this
definition in mind,
%in the most important case of the {\em orthogonal} separation of variables
the separability conditions for the Schr\"odinger equation, in addition to
those for the Hamilton--Jacobi equation, include the so-called Robertson
condition, see \cite{Ben3, Ben4} and Theorem~\ref{sept} below for details.

It is well known that separation of variables is intimately related to
complete integrability (and also to superintegrability, see e.g.\ \cite{wint,
wint2, mkk} and references therein) and to existence of symmetries, see
e.g.\ \cite{bg, mi, ka} and references therein. As an aside, note that another
important hallmark of complete integrability, the bihamiltonian property,
under certain assumptions implies separability \cite{pe}, and conversely, any
separable Hamiltonian system under certain conditions admits a bihamiltonian
representation, possibly in the extended phase space \cite{bla1}.\looseness=-1

The separation of variables implies complete integrability but there are
examples of completely integrable (or even superintegrable) systems for which
the separation variables are not known or the separation of variables does not
occur at all, see e.g.\ \cite{Ber} and references therein; cf.\ e.g.\ also
\cite{RW} and references therein regarding the search for separation coordinates.

The separation of variables in the Schr\"odinger equation implies
commutativity of the quantized versions of the integrals of motion resulting
from the separation relations, see Theorem~\ref{pri} below and \cite{Ben4} for
details. It turns out, however, that these quantized integrals of motion
commute even under somewhat milder conditions than the separation of variables
in the Schr\"odinger equation. Namely, it suffices to require the separation
of variables in the Hamilton--Jacobi equation and the so-called pre-Robertson
condition (instead of the stronger Robertson condition, cf.\ the discussion
above), see \cite{Ben4} and Theorem~\ref{pri} below for details.
%Moreover, there indeed exist
%examples of quantum (super)integrable systems that do not admit separation
%of variables \cite{Ber}.

These facts beget a deeper analysis of the interplay of separability and
quantum complete integrability, and the present paper is just the first step
in this direction. We study here a class of integrable systems for which the
entries of the St\"ackel matrix are monomials in the separation variables and
show that for this class the Robertson and pre-Robertson conditions are
equivalent, and the only subclass for which these are satisfied is, up to a
natural equivalence, that of the so-called Benenti systems, see Section 4
below for details. Note that this subclass contains a lot of systems naturally
arising in physics and mechanics. Our result is general and contains the proof
of quantum integrability and quantum separability, in the sense of \cite{Ben3,
Ben4}, for the whole class of Benenti systems. Up to now, only particular
members of that class was considered by various authors \cite{ZL, Toth, HW,
Du}.

In what follows we restrict ourselves to quantization of a special class of
the St\"ackel systems for which all separation coordinates are essential,
i.e., the St\"ackel matrix depends on all of them in a nontrivial fashion.
Thus, we deal with the so-called strictly orthogonal separation. As a
consequence of this, at the quantum level we have the associated free quantum
separability, where the eigenfunctions of the Hamiltonian operator take the
form of a product where each term depends on a single separation coordinate.
The case of non-orthogonal separability, also known as reduced or constrained
separability \cite{Ben3}, when the so-called ignorable coordinates come up,
which in turn leads to $R$-separability instead of free separability at the
quantum level, requires a separate study.

The paper is organized as follows. In Section~2 we recall some basic results
on the separation of variables in the Hamilton--Jacobi equation and the
Schr\"odinger equation. In Section~3 we introduce the class of integrable
systems under study and establish the equivalence of the Robertson and
pre-Robertson conditions for this class. In Section~4 we show that the only
systems within the class under study that satisfy the pre-Robertson condition
are the so-called systems of Benenti type. In Section~5 we illustrate our
results by a simple example. Finally, in Section~6 we make some comments on
the existing ambiguities in the quantization procedure and their influence on
quantum integrability and quantum separability.

\section{Preliminaries}

Consider a quadratic in momenta Hamiltonian function in the natural form
\begin{align}
\label{nh}H = \frac{1}{2}g^{ij}(q)p_{i}p_{j} + V(q),\qquad p\in T^{*}%
_{q}\mathcal{Q},
\end{align}
on the phase space $T^{*}\mathcal{Q}$, the cotangent space of an
$n$-dimensional (pseudo-)Riemannian manifold $(\mathcal{Q},g)$, where $V$ is a
function on $\mathcal{Q}$ (the potential), and the sum over repeated indices
from $1$ to $n$ is understood unless otherwise explicitly stated. The related
Hamilton--Jacobi equation has the form
\begin{align}
\label{hj}\frac{1}{2}g^{ij} \partial_{i}S \partial_{j} S + V = E,
\end{align}
where $E$ is a constant parameter (energy) and $\partial_{i}=\partial/\partial
q^{i}$. The corresponding stationary Schr\"odinger equation is\looseness=-1
\begin{align}
\label{se}\hat{H}\psi:= -\frac{1}{2}\hslash^{2} g^{ij}\nabla_{i}\nabla_{j}
\psi+ V\psi=E\psi,
\end{align}
where $\nabla_{i}$ is the covariant derivative for the Levi-Civita connection
of $g$ and $\hslash$ is a parameter (the Planck constant).

\begin{theorem}
\textrm{(\cite{Ben3})}\label{sept} The Schr\"odinger equation is freely
separable in a coordinate system if the following conditions hold: the
coordinates in question are orthogonal, the corresponding Hamilton--Jacobi
equation is separable, and in these coordinates the Robertson condition is
satisfied:
\begin{align}
\label{rc}R_{ij} = \frac{3}{2}\partial_{i}\Gamma_{j} = 0,\qquad i\neq j;
\end{align}
here $R_{ij}$ is the Ricci tensor and the contracted Christoffel symbols are
defined by
\begin{align*}
\Gamma_{i} := g_{is}g^{jk}\Gamma^{s}_{jk} = \frac{1}{2}g^{jk}\! \left(
\partial_{j}g_{ki}+\partial_{k}g_{ij}-\partial_{i}g_{jk}\right) .
\end{align*}

\end{theorem}

In the orthogonal coordinates we have
\begin{align}
\label{go}\Gamma_{i} = \frac{1}{2}\partial_{i} \log\left| \det G\right|  -
g_{ij}\partial_{k} g^{jk},\qquad G^{jk}\equiv g^{jk}.
\end{align}

A \textit{Killing--St\"ackel algebra} is \cite{Ben2} an $n$-dimensional linear
space spanned by contravariant Killing tensors $K_{r}$ of valence two which
can be simultaneously diagonalized in orthogonal coordinates. Such an algebra
naturally contains the contravariant metric $K_{1}\equiv G$. With a
Killing--St\"ackel algebra we can associate a system of $n$ Hamiltonians
\begin{align}
\label{hk}H_{r} = \frac{1}{2}K_{r}^{ij}p_{i}p_{j} + V_{r},\qquad r=1,\ldots,n,
\end{align}
where $V_{r}$ are functions on $\mathcal{Q}$ and $V_{1}\equiv V$. Then $H_{1}$
is nothing but the original Hamiltonian \eqref{nh}.

\begin{theorem}
\textrm{(\cite{Ben2})} The Hamilton--Jacobi equation \eqref{hj} associated
with a natural Hamiltonian \eqref{nh} is separable in orthogonal coordinates,
i.e., integrable by separation of variables if and only if there exists a
Killing--St\"ackel algebra such that the equation
\begin{align*}
d(\bar{K}_{r}dV) = 0
\end{align*}
holds for all $r=1,\ldots,n$, where $\bar{K}_{r} = gK_{r}$ (i.e., $(\bar
{K}_{r})^{i}_{j} = g_{js}(K_{r})^{si}$). Then there exist the functions
$V_{r}$ on $\mathcal{Q}$ satisfying
\begin{align*}
dV_{r} = \bar{K}_{r}dV,\qquad r=1,\ldots,n,
\end{align*}
such that the associated Hamiltonians \eqref{hk} Poisson commute, $\left\{
H_{i},H_{j}\right\} =0$.
\end{theorem}

The separation relations \cite{Sklyanin} associated with separable
Hamilton--Jacobi equations generated by the Hamiltonians of the form
\eqref{hk} are
\begin{align}
\label{gsr}\sum_{r=1}^{n} S_{i}^{r}(\lambda_{i}) H_{r} = f_{i}(\lambda_{i}%
)\mu_{i}^{2} + \sigma_{i}(\lambda_{i}),\qquad i=1,\ldots,n,
\end{align}
where $\lambda$ and $\mu$ are orthogonal coordinates on $\mathcal{Q}$ and the
associated momenta, respectively and $S(\lambda)$ is called a St\"ackel
matrix. The relations \eqref{gsr} are the (original) St\"ackel separation
relations quadratic in the momenta and the associated dynamical systems are
the related St\"ackel separable systems \cite{St, bla1}. The functions
$S_{i}^{r}$, $f_{i}$ and $\sigma_{i}$ are functions of a single argument
$\lambda_{i}$ which are uniquely determined by the Killing tensors $K_{r}$ and
the potentials $V_{r}$. On the other hand, one can start with the separation
relations \eqref{gsr} and generate separable natural Hamiltonian systems with
the Hamiltonians~\eqref{hk}.

Introduce (cf.\ e.g.\ \cite{Ben3, Ben4} and references therein) linear
second-order differential operators corresponding to the Hamiltonians
\eqref{hk}:
%of the form%
\begin{align}
\label{qhk}\hat{H}_{r} = -\frac{1}{2}\hslash^{2}\nabla_{i} K_{r}^{ij}%
\nabla_{j} + V_{r},\qquad k=1,\ldots,n.
\end{align}
Thus, $\hat{H}_{1}$ coincides with the operator $\hat{H}$ defining the
Schr\"odinger equation \eqref{se}. In general, these operators do not
necessarily commute even when they are associated to some Killing--St\"ackel algebra.

\begin{theorem}
\textrm{(\cite{Ben4})}\label{pri} Let the Hamiltonians \eqref{hk} form the
space of first integrals in involution associated with the orthogonal
separation for the Hamilton--Jacobi equation. Then the corresponding
Hamiltonian operators \eqref{qhk} commute, that is, $[\hat{H}_{i},\hat{H}%
_{j}]=0$, if and only if the pre-Robertson condition
\begin{align}
\label{prc}\partial_{i} R_{ij} - \Gamma_{i} R_{ij} = 0,\qquad i\neq
j,\quad\mbox{\rm no sum over}\, i,
\end{align}
is satisfied in any orthogonal separable coordinates.
\end{theorem}

Since the Ricci tensor is symmetric, in the orthogonal separable coordinates
the condition \eqref{prc} takes the form
\begin{align}
\label{prc2}\partial_{j}\! \left( \partial_{i}\Gamma_{i} - \frac{1}{2}%
\Gamma_{i}^{2}\right)  = 0,\qquad i\neq j,\quad\mbox{\rm no sum over}\, i,
\end{align}
where we used the fact that $R_{ij} = \frac{3}{2}\partial_{i}\Gamma_{j}$ in
these coordinates.

In the present paper we will consider quantization of St\"ackel systems
associated to a particular class of separation relations \eqref{gsr} with
$S_{i}^{r}(\lambda_{i})$ being monomials, namely, the relations of the form
\eqref{Stack}. We will show that in this case:

\begin{itemize}
\item the pre-Robertson condition \eqref{prc} is equivalent to the Robertson
condition \eqref{rc};

\item hence the related Hamiltonian operators \eqref{qhk} pairwise commute if
and only if the related Schr\"odinger equation is separable;

\item the only class which satisfies the Robertson condition \eqref{rc} is the
Benenti class with $S_{i}^{r}=\lambda_{i}^{n-r}$.
\end{itemize}

The Benenti class is an important case of the separation relations \eqref{gsr}
which has the form \eqref{ben}, see Section 4 below for details.

\section{Classical St\"{a}ckel systems}

Consider a classical St\"{a}ckel system involving $n$ Hamiltonians $H_{i}$
that originate from a set of $n$ separation relations of the form
\begin{equation}
\label{Stack}\sum_{r=1}^{n} H_{r}\lambda_{i}^{\delta_{r}} = f_{i}(\lambda
_{i})\mu_{i}^{2} + \sigma_{i}(\lambda_{i}),\qquad i=1,\ldots,n,
\end{equation}
where $f_{i}$ and $\sigma_{i}$ are arbitrary functions of one argument and
where all $\delta_{i}\in\mathbb{Z}$, $i=1,\ldots,n$, are pairwise distinct.
This is a special case of relations (\ref{gsr}) with a particular choice of
St\"ackel matrix, i.e. $S^{r}_{i}=\lambda^{\delta_{r}}_{i}$.

Without loss of generality we can assume the following ordering:
\begin{align}
\label{Stack1}\delta_{1} > \delta_{2} > \ldots>\delta_{n-1}>\delta_{n} = 0.
\end{align}
We adopt the normalization $\delta_{n}=0$ since we always can divide the left
and right-hand sides of the separation relations \eqref{Stack} by $\lambda
_{i}^{\delta_{n}}$ while preserving their form. Thus, fixing a sequence
$(\delta_{1},\delta_{2},\dots,\delta_{n})$ we can choose a class of
St\"{a}ckel systems. An interested reader can find further particulars on the
classification of generalized St\"{a}ckel systems in \cite{bla1}.

The separation relations \eqref{Stack} constitute a system of $n$ equations
linear in the unknowns $H_{r}$. Solving these relations with respect to
$H_{r}$ we obtain, on the phase space $T^{*}\mathcal{Q}$, $n$ commuting
St\"{a}ckel Hamiltonians of the form
\begin{equation}
H_{r}=\mu^{T}\bar{K}_{r}G\mu+V_{r}(\lambda),\qquad r=1,\ldots,n,
\label{Stackham}%
\end{equation}
where $\lambda=(\lambda_{1},\ldots,\lambda_{n})^{T}$ are the orthogonal
separation coordinates on $\mathcal{Q}$ and $\mu=(\mu_{1},\ldots,\mu_{n})^{T}$
are the associated momenta (here and below the superscript $T$ indicates the
transposed matrix). By definition $\bar{K}_{1} = I$, where $I$ is the unit
matrix. The objects $\bar{K}_{r}$ in (\ref{Stackham}) can be interpreted as
Killing tensors of the type $(1,1)$ on $\mathcal{Q}$ for the (contravariant)
metric $G$. The metric tensor $G$ and all the Killing tensors $\bar{K}_{r}$
are diagonal in the $\lambda$ variables. Then the contravariant tensors $K_{r}
= \bar{K}_{r}G$ of the type $(2,0)$ form a Killing--St\"ackel algebra.

The relations (\ref{Stack}) can be written in the matrix form as
\[
SH=U
\]
where $H=(H_{1},\ldots,H_{n})^{T}$ and $U$ is a St\"{a}ckel vector,
\begin{align*}
U=(f_{1}(\lambda_{1})\mu_{1}^{2}+\sigma_{1}(\lambda_{1}),\ldots,f_{n}%
(\lambda_{n})\mu_{n}^{2}+\sigma_{n}(\lambda_{n}))^{T},
\end{align*}
while $S$ is a classical St\"{a}ckel matrix,
\begin{align}
\label{smatrix}S=
\begin{pmatrix}
\lambda_{1}^{\delta_{1}} & \cdots & \lambda_{1}^{\delta_{n-1}} & 1\\
\vdots & \ddots &  & \vdots\\
\lambda_{n}^{\delta_{1}} & \cdots & \lambda_{n}^{\delta_{n-1}} & 1
\end{pmatrix}
.
\end{align}
Note that our assumption that no $\delta_{i}$ coincide implies that
$\det(S)\neq0$. Thus, the Hamiltonians (\ref{Stackham}) can be represented in
the matrix form as $H=S^{-1}U$, which also means that the metric $G$ in
(\ref{Stackham}) can be written as%
\begin{align*}
G= \diag \left(  f_{1}(\lambda_{1})\left(  S^{-1}\right)  _{11},\ldots
,f_{n}(\lambda_{n})\left(  S^{-1}\right)  _{1n}\right)  ,
\end{align*}
and thus the Killing tensors $\bar{K}_{r}$ in (\ref{Stackham}) read%
\begin{align*}
\bar{K}_{r}= \diag\left(  \left(  S^{-1}\right)  _{r1}/\left(  S^{-1}\right)
_{11},\ldots,\left(  S^{-1}\right)  _{rn}/\left(  S^{-1}\right)  _{1n}\right)
,\qquad r=1,\ldots,n.
\end{align*}

Notice that all St\"{a}ckel systems constructed from the separation relations
(\ref{Stack}) can be divided into various classes \cite{bla1}. A given class
is distinguished by fixing the sequence of natural numbers (\ref{Stack1}),
i.e., a St\"{a}ckel matrix $S$ (\ref{smatrix}). Within a given class the
functions $f_{i}(\lambda_{i})$ parametrize the admissible set of metrics $G$
related to $S$ which share the same set of Killing tensors $\bar{K}_{r}$ while
$\sigma(\lambda_{i})$ parametrize separable potentials.

Using \eqref{go} we can express the contracted Christoffel symbols for the
metric $G$ in the form
\begin{equation}
\label{Gi}\Gamma_{i} = \frac{1}{2}\partial_{i} \log F_{i},
\end{equation}
where
\begin{equation}
\label{Fi}F_{i} = \frac{\prod_{k \neq i} g^{kk}}{g^{ii}} = \frac{\prod_{k \neq
i} (S^{-1})_{1k}}{(S^{-1})_{1i}} \frac{\prod_{k \neq i} f_{k}(\lambda_{k}%
)}{f_{i}(\lambda_{i})} \equiv\frac{\gamma_{i}(\lambda)}{D_{i1}(\lambda)}
\frac{\prod_{k \neq i} f_{k}(\lambda_{k})}{f_{i}(\lambda_{i})}.
\end{equation}
The symbol $D_{i1}(\lambda)$ stands for the $(i,1)$-cofactor of $S$, which is
$\lambda_{i}$-independent and
\begin{align*}
\gamma_{i}(\lambda) := \frac{\prod_{k \neq i} D_{k1}}{(\det S)^{n-2}}.
\end{align*}

We see that $\gamma_{i}$ is a quotient of two polynomials in $\lambda_{i}$,
\begin{equation}
\label{gi}\gamma_{i}(\lambda) = \frac{A_{i}}{B_{i}},
\end{equation}
where
\begin{align*}
A_{i} = \prod_{k \neq i} D_{k1},\qquad D_{k1} =a_{2}^{ik} \lambda_{i}%
^{\delta_{2}} + \ldots+ a_{n-1}^{ik} \lambda_{i}^{\delta_{n-1}} + a_{n}^{ik},
\qquad\mbox{\rm no sum over}\, i,
\end{align*}
and
\begin{align*}
B_{i}= (\det S)^{n-2},\qquad\det S = b^{i}_{1} \lambda_{i}^{\delta_{1}} +
\ldots+ b^{i}_{n-1} \lambda_{i}^{\delta_{n-1}} + b^{i}_{n},\qquad
\mbox{\rm no sum over}\, i.
\end{align*}
The coefficients $a_{j}^{ik}$ and $b^{i}_{j}$ are polynomials in all remaining
variables $\lambda_{s}$ ($s \neq i$).

\begin{lemma}
\label{lrc} The Robertson condition \eqref{rc}
%\eq{\label{rc}
%\partial_j \Gamma_i = 0 \qquad i \neq j
%}
holds if and only if for each $i=1,\dots,n$ the function $\gamma_{i}$ is
$\lambda_{i}$-independent.\looseness=-1
\end{lemma}

\begin{proof}
Upon using \eqref{Gi} the Robertson condition \eqref{rc}
boils down to
\eq{\label{rcg}
\pr_i\pr_j\log F_i = \pr_i\pr_j \log \gamma_i=0,\qquad i\neq j,\qquad\mbox{\rm no sum over}\, i,
}
which can hold if and only if all $\gamma_i$ factorize as
\eq{\label{fact}
\gamma_i(\la) = \phi_i(\la_i)\psi_i(\la_1,\dots,\la_{i-1},\la_{i+1},\dots,\la_{n}),
}
i.e., $\phi_i(\la_i)$ is a function of $\la_i$ alone and $\psi_i$
is independent of $\la_i$.
The determinant of the St\"ackel matrix \eqref{smatrix}
is a homogeneous function of the coordinates $\la_i$ of degree $\delta = \delta_1 + \ldots + \delta_{n-2} + \delta_{n-1}$, that is,
\eqq{
\det S (\kappa \la) = \kappa^\delta \det S(\la).
}
Each cofactor $D_{k1}$ is a homogeneous function of degree $\delta - \delta_1$. Therefore,
the coefficients $a_j^{ik}$ and $b^i_j$ in the factorization \eqref{gi}
are all nonzero homogeneous polynomials in $\la_s$ ($s\neq i$) of degrees
$\delta - \delta_1 - \delta_j$ and $\delta - \delta_j$, respectively.
These degrees are different by virtue of our assumption that $\delta_i$ are pairwise distinct,
and thus the quantities $a_j^{ik}$ and $b^i_j$ are linearly independent as functions
for any fixed $i$. This means that
$D_{k1}$ and $\det S$ as polynomials in $\la_i$ cannot have constant zeros independent
of all other coordinates $\la_s$. Thus, $\gamma_i$ as a rational function in $\la_i$
cannot have constant zeros or poles. Hence $\gamma_i$ cannot factorize
in the fashion given by \eqref{fact} unless all $\phi_i$ are constants, that is, for each $i=1,\dots,n$
the function $\gamma_i$ is independent of $\la_i$.
\end{proof}

\begin{theorem}
For the class of St\"ackel systems related to separation relations
\eqref{Stack}, the pre-Robertson condition \eqref{prc} is satisfied  if and
only if the Robertson condition \eqref{rc} holds.
\end{theorem}

\begin{proof}
Using \eqref{prc2} and \eqref{Gi} turns the pre-Robertson condition \eqref{prc}
in the orthogonal coordinates $\la$ into
\eqq{
\pr_j\bra{\pr_i\log F_i - \frac{1}{4}(\log F_i)^{\,2}}=0,\qquad i\neq j,\qquad\mbox{\rm no sum over}\, i,
}
and employing (\ref{Fi}) for $F_i$ yields
\eq{\label{prg}
\pr_j\bra{\frac{\pr_i^2\gamma_i}{\gamma_i}-\frac{5}{4}\frac{(\pr_i\gamma_i)^2}{\gamma_i^2}
+\frac{1}{2}\frac{\pr_i\gamma_i}{\gamma_i}\frac{\pr_i f_i}{f_i}} = 0,\qquad i\neq j, \qquad\mbox{\rm no sum over}\, i.
}
Since $\gamma_i$ is a homogeneous rational function, the first and the second term in the bracket in \eqref{prg}
are also homogeneous functions of degree $-2$. The function $f_i$ depends only on $\la_i$. Thus, the third term, $\frac{\pr_i\gamma_i}{\gamma_i}\frac{\pr_i f_i}{f_i}$, can have the same degree of homogeneity as the other two
only if $\frac{\pr_i f_i}{f_i}$ is proportional to $\la_i^{-1}$. This would mean that the third term has a pole at
$\la_i=0$. Using the results of the analysis of properties of $\gamma_i$ in the proof of Lemma~\ref{lrc}, we see that
the first and the second term cannot have a constant pole $\la_i=0$. Therefore, \eqref{prg} can hold only if
\eqq{
\pr_j\bra{\frac{\pr_i^2\gamma_i}{\gamma_i}-\frac{5}{4}\frac{(\pr_i\gamma_i)^2}{\gamma_i^2}} = 0\quad\mbox{\rm and}\quad
\pr_j\bra{\frac{\pr_i\gamma_i}{\gamma_i}\frac{\pr_i f_i}{f_i}} = 0,\qquad i\neq j,\qquad\mbox{\rm no sum over}\, i.
}
The second of the above conditions is equivalent to
\eqq{
\pr_j\bra{\frac{\pr_i\gamma_i}{\gamma_i}} = 0\qquad\iff\qquad  \pr_i\pr_j \log \gamma_i=0,\qquad i\neq j,\qquad\mbox{\rm no sum over}\, i,
}
which is nothing but \eqref{rcg}, i.e., the pre-Robertson condition indeed reduces to the Robertson condition.
\end{proof}

\begin{lemma}
\label{lem} The equality
\begin{align}
\label{nc}\delta_{k} = (n - k)\delta_{n-1},
\end{align}
where $k=1,\ldots,n-1$, is a necessary condition for the functions $\gamma
_{i}(\lambda)$ to be independent of $\lambda_{i}$ and thus for the Robertson
condition \eqref{rc} to hold.
\end{lemma}

\begin{proof}
Fix $i$ and assume that the function $\gamma_i(\la)$ from \eqref{gi} is independent of $\la_i$. This means that
the polynomials $A_i$ and $B_i$
must be proportional to each other, with the proportionality factor being a $\la_i$-independent function
which can, however, depend on $\lambda_j$ for all $j\neq i$. This means that the polynomials $A_i$ and $B_i$ must contain
the identical sets of powers of $\lambda_i$.  We will show that this implies the condition
\eq{\label{rel}
\delta_{j} = \frac{n-j}{n-1}\delta_1\qquad j=2,\ldots,n-1,
}
which is equivalent to (\ref{nc}).
Bearing in mind that all powers $\delta_j$ are natural numbers
such that $\delta_{j}-\delta_{j+1}\me 1$,
$\delta_{j}\me n-j$, we can order the powers of $\lambda_i$ in $A_i$ and $B_i$ as follows.
{\postdisplaypenalty=10000
First, the $n-2$ highest powers of $\lambda_i$ in $A_i$ can be arranged into
the decreasing sequence
\eqq{
\begin{array}{llllllllll}
\lambda_i^{(n-1)\delta_2} & \lambda_i^{(n-2)\delta_2 + \delta_3} &
\lambda_i^{(n-3)\delta_2 + 2\delta_3} &
\lambda_i^{(n-4)\delta_2 + 3\delta_3} & & \la_i^{(n-k)\delta_2 + (k-1)\delta_3} \\
& & \lambda_i^{(n-3)\delta_2 + \delta_2 + \delta_4} &
\lambda_i^{(n-4)\delta_2 + \delta_2 + \delta_3 + \delta_4} & \cdots &\qquad \vdots & \cdots \\
& & & \lambda_i^{(n-4)\delta_2 + 2\delta_2 + \delta_5} & & \la_i^{(n-k)\delta_2 + (k-2)\delta_2 +  \delta_{k+1}}
\end{array},
}
where each group is numbered by $k=1,2,\ldots,n-2$,
the $k^{\text{th}}$ group has the form
\eqq{
\la_i^{p_k(r)}:=    \la_i^{(n-k)\delta_2 + \delta_{r_1} +
\ldots + \delta_{r_{k-1}}}\qquad r_1 +  \ldots + r_{k-1} = 3(k-1),
}
and $p_k(r)$ satisfies $p_k(r)\me n^2-3n-k+3$.}\looseness=-1
Likewise, we have the following decreasing sequence of the $n-2$ highest powers of $\lambda_i$ contained in $B_i$:
\eqq{
\begin{array}{lllllllll}
\lambda_i^{(n-2)\delta_1} & \lambda_i^{(n-3)\delta_1 + \delta_2} &
\lambda_i^{(n-4)\delta_1 + 2\delta_2} &
\lambda_i^{(n-5)\delta_1 + 3\delta_2} & & \la_i^{(n-k-1)\delta_1 +(k-1)\delta_2}\\
& & \lambda_i^{(n-4)\delta_1 + \delta_1 + \delta_3} &
\lambda_i^{(n-5)\delta_1 + \delta_1 + \delta_2 + \delta_3} & \cdots &\qquad \vdots & \cdots  \\
& & & \lambda_i^{(n-5)\delta_1 + 2\delta_1 + \delta_4} & & \la_i^{(n-k-1)\delta_1 + (k-2)\delta_1 + \delta_{k}}
\end{array};
}
for $k=1,\ldots,n-2$ we have
\eqq{
\la_i^{q_k(s)}:= \la_i^{(n-k-1)\delta_1 + \delta_{s_1} + \ldots + \delta_{s_{k-1}}}\qquad s_1 +  \ldots + s_{k-1} = 2(k-1)
}
and $q_k(s)$ satisfies $q_k(s)\me n^2-3n-k+3$.
The sequences of $n-2$ highest powers of $\lambda_i$ in $A_i$ and $B_i$ must be identical. Taking into account
the above orderings of powers of $\lambda_i$ and equating the first two elements of the sets in question
we find that
$p_1=q_1$ and $p_2=q_2$. Hence, $\delta_2 = \frac{n-2}{n-1}\delta_1$ and $\delta_3 =  \frac{n-3}{n-1}\delta_1$ in accordance with \eqref{rel}. Continuing the process by induction, we see that the remaining $\delta_j$ will have to be
of the form
\eqq{
\delta_j = \frac{n-m_j}{n-1}\delta_1,
}
where all $m_j$ are integers.  We must require that for each $j$ we have $n>m_{j}>m_{j-1}$ in order that $\delta_j>\delta_{j-1}$
and $\delta_j\neq 0$ for $j\neq n$.
Thus, these inequalities can hold only if $m_j =j$, that is, \eqref{rel} holds.
In this case each of the above groups in $A_i$ and $B_i$ boils down to a single monomial of the degree
$p_k = q_k = n^2-3n-k+3$.
\end{proof}

\section{The Benenti class}

An important class of St\"{a}ckel systems is obtained by setting $\delta
_{i}=n-i$. Then the separation relations \eqref{Stack} take the form
\begin{align}
\label{ben}\sum_{r=1}^{n} H_{r}\lambda_{i}^{n-r}=f_{i}(\lambda_{i})\mu_{i}%
^{2}+\sigma_{i}(\lambda_{i})\qquad i=1,\ldots,n.
\end{align}
The resulting Hamiltonians $H_{i}$ constitute a completely integrable system
that we will call a Benenti system to honor the fundamental contributions
\cite{Ben1,Ben2} of S. Benenti to the study of these objects. The systems in
question enjoy a number of remarkable properties.
%at the classical and quantum level.
For instance, among them we find plenty of superintegrable systems \cite{BS05}
which are exactly solvable in both classical and quantum mechanics \cite{S}. \looseness=-1

Note that if $f_{i}$ and $\sigma_{i}$ are the same for all $i$, i.e.,
$f_{i}=f(\lambda_{i})$ and $\sigma_{i}=\sigma(\lambda_{i})$, then the
relations (\ref{ben}) are nothing but $n$ copies of a single separation curve
\begin{align}
\label{ben0}\sum_{r=1}^{n} H_{r}\lambda^{n-r}=f(\lambda)\mu^{2}+\sigma
(\lambda).
\end{align}

For the Benenti systems it is possible to give compact formulas for many
objects introduced above. In particular, the St\"ackel matrix $S$ is the
Vandermonde matrix, $S_{ij} = \lambda_{i}^{n-j}$, with the determinant
\begin{align*}
\det S = \prod_{1\leqslant i<j\leqslant n} (\lambda_{j}-\lambda_{i}).
\end{align*}
The metric $G$ and the Killing tensors $\bar{K}_{r}$ in \eqref{Stackham} are
given explicitly by the formulas
\begin{align*}
G=\diag \! \left( \frac{f_{1}(\lambda_{1})}{\Delta_{1}},\ldots,\frac
{f_{n}(\lambda_{n})}{\Delta_{n}}\right) ,\quad\bar{K}_{r} & =-\diag\left(
\frac{\partial\rho_{r}}{\partial\lambda_{1}},\ldots,\frac{\partial\rho_{r}%
}{\partial\lambda_{n}}\right) ,\quad\Delta_{i}=\prod_{j\neq i} (\lambda
_{i}-\lambda_{j}),
\end{align*}
where $r=1,\ldots,n$. Here $\rho_{r}=\rho_{r}(\lambda)$ are the Vi\`{e}te
polynomials (symmetric polynomials with the sign factors) in the $\lambda$
variables:%
\begin{align*}
\rho_{r}(\lambda)=(-1)^{r} \sum_{1\leqslant s_{1}<\ldots<s_{r}\leqslant n}
\lambda_{s_{1}}\cdot\ldots\cdot\lambda_{s_{r}}\qquad r=1,\ldots,n.
\end{align*}

A compact formula (first presented in \cite{macartur2011}) for the separable
potentials $V_{r}^{(k)}$ related to $\sigma_{i}(\lambda_{i})=\lambda_{i}^{k}$
\eqref{ben} reads
\begin{equation}
V^{(k)}=F^{k}V^{(0)},\qquad k\in\mathbb{Z},\label{Vk}%
\end{equation}
where $V^{(k)}=(V_{1}^{(k)},\ldots,V_{n}^{(k)})$, $V^{(0)}=(0,\ldots,0,1)$
and
\begin{equation}
F=\left(
\begin{array}
[c]{cccc}%
-\rho_{1} & 1 &  & \\
-\rho_{2} &  & \ddots & \\
\vdots &  &  & 1\\
-\rho_{n} & 0 & \cdots & 0
\end{array}
\right) .\label{Fmat}%
\end{equation}

Let us stress again that the Killing tensors $\bar{K}_{r}$ do not depend on a
particular choice of $f_{i}$ and $\sigma_{i}$. It can be shown that as long as
the functions $f_{i}$ are the same for all $i$, i.e., $f_{i}=f(\lambda_{i})$,
and $f$ is a polynomial of degree less than $n+1$, then the metric $G$ is flat
(the explicit formulas for its flat coordinates can be found in \cite{BS}, and
the quantization in these coordinates is discussed in \cite{S}), while if $f$
is a polynomial of degree $n+1$ then $G$ has constant but nonvanishing
curvature. It is the reason why this particular class (\ref{ben}) of the
St\"ackel systems includes a majority of known separable systems of classical
mechanics. For all other classes of St\"ackel systems given by different
sequences \eqref{Stack1} related metrics are neither flat nor of constant curvature.

\begin{proposition}
The Benenti class \eqref{ben} satisfies the Robertson condition \eqref{rc}.
\end{proposition}

\begin{proof}
Recall that in the case under study each cofactor $D_{i1}$ is independent of $\la_i$.
We have
\eqq{
\frac{\gamma_i}{D_{i1}} = (-1)^n\frac{\Delta_i}{\prod_{k\neq i} \Delta_k},
}
where on the right-hand side all terms involving $\la_i$ cancel. Hence,
each $\gamma_i$ is also independent of $\la_i$, and the proposition follows from Lemma~\ref{lrc}.
\end{proof}

\begin{theorem}
\label{thr} The only class of the St\"ackel systems associated with the
separation relations \eqref{Stack} which satisfies the Robertson condition
\eqref{rc} is, up to a natural equivalence, the Benenti class \eqref{ben}.\looseness=-1
\end{theorem}

\begin{proof}
By Lemma~\ref{lem} the Robertson condition can hold only if
\eqref{nc} is valid. However, this case of the St\"ackel systems related to  \eqref{Stack} is
equivalent to the Benenti case \eqref{ben} via the canonical transformation
\eqq{
\lambda_i\map  \lambda_i^{\delta_{n-1}}\qquad
\mu_i\map \frac{1}{\delta_{n-1}} \lambda_i^{1 - \delta_{n-1}} \mu_i.
}
As a result, the separation relations \eqref{Stack} with $\delta_i$ given by \eqref{rel}
take the form \eqref{ben} in the new separation coordinates.
\end{proof}

The Benenti class contains many known separable systems from classical
mechanics, all of which are also separable in the quantum case, as it was
established earlier in this section. In particular, for the Euclidean case
there exists an infinite family of potentials separable in generalized
elliptic coordinates \cite{AW}, containing the well-known Garnier system
\cite{G} (see the example from the subsequent section). Next, there is an
infinite family of potentials separable in generalized parabolic coordinates
\cite{Stefan}. For the pseudo-Euclidean case we have an infinite family of
separable potentials considered in \cite{BS} which include \emph{inter alia}
stationary flows of the coupled Korteweg--de Vries and coupled Harry Dym
soliton systems \cite{BM1, BM2}. Finally, for the constant curvature case,
there exists another infinite family of potentials separable in generalized
spherical-conical coordinates, which contains, in particular, the
Neumann--Rosochatius potential \cite{NR}.

It is also important to note that all other classes of St\"ackel systems
considered in the present paper, for which we proved that quantum
integrability does not survive, are not independent from Benenti class. In
fact, all remaining classes of St\"ackel systems (\ref{Stack}) are related to
the Benenti class by multi-parameter generalized St\"ackel transforms at the
level of Hamiltonians and by the so-called reciprocal transformations at the
level of equations of motion \cite{bla2,bla3,S12}.

\section{Example}

As we have noted earlier, the Benenti class contains many separable systems
known from classical mechanics. Here we illustrate our results on a simple
example of a system with two degrees of freedom, namely, the two-dimensional
Garnier system \cite{G}. In the Euclidean coordinates $(x_{1},x_{2})$ on
$\mathbb{R}^{2}$ the Hamiltonian $H$ and the second constant of motion $F$
read
\[
H=\frac{1}{2}p_{1}^{2}+\frac{1}{2}p_{2}^{2}+\frac{1}{16}(x_{1}^{2}+x_{2}%
^{2})^{2}-\frac{1}{4}(\beta_{1}x_{1}^{2}+\beta_{2}x_{2}^{2})-\beta_{1}%
\beta_{2},
\]
\[
F=\frac{1}{2}(-\beta_{2}+\frac{1}{4}x_{2}^{2})p_{1}^{2}+\frac{1}{2}(-\beta
_{1}+\frac{1}{4}x_{1}^{2})p_{2}^{2}-\frac{1}{4}x_{1}x_{2}p_{1}p_{2}- \frac
{1}{16}(x_{1}^{2}+x_{2}^{2})(\beta_{2}x_{1}^{2}+\beta_{1}x_{2}^{2})+\frac
{1}{4}\beta_{1}\beta_{2}(x_{1}^{2}+x_{2}^{2}),
\]
where $\beta_{1}\neq\beta_{2}\in\mathbb{R}$. The Garnier potential is the
simplest nontrivial potential that separates in generalized elliptic
coordinates \cite{AW}. In our case the separation coordinates $(\lambda
_{1},\lambda_{2})$ are elliptic coordinates related to the Euclidean ones by
the formulas
\[
x_{1}^{2}=4\frac{(\beta_{1}-\lambda_{1})(\beta_{1}-\lambda_{2})}{(\beta
_{1}-\beta_{2})},\;\;\;x_{2}^{2}=4\frac{(\beta_{2}-\lambda_{1})(\beta
_{2}-\lambda_{2})}{(\beta_{2}-\beta_{1})}.
\]
The separation relations are given by two copies of the separation curve
\[
H\lambda+F=-\frac{1}{2}(\lambda-\beta_{1})(\lambda-\beta_{2})\mu^{2}%
-(\beta_{1}+\beta_{2})\lambda^{2}+\lambda^{3},\;\;\;\lambda=\lambda
_{1},\lambda_{2}.
\]
In the separation coordinates we have
\[
H=-\frac{1}{2}\frac{(\lambda_{1}-\beta_{1})(\lambda_{1}-\beta_{2})}%
{\lambda_{1}-\lambda_{2}}\mu_{1}^{2} -\frac{1}{2}\frac{(\lambda_{2}-\beta
_{1})(\lambda_{2}-\beta_{2})}{\lambda_{2}-\lambda_{1}}\mu_{2}^{2}+(\beta
_{1}+\beta_{2})(\lambda_{1}+\lambda_{2}) +\lambda_{1}\lambda_{2}-(\lambda
_{1}+\lambda_{2})^{2},
\]
\[
F=\frac{1}{2}\frac{\lambda_{2}(\lambda_{1}-\beta_{1})(\lambda_{1}-\beta_{2}%
)}{\lambda_{1}-\lambda_{2}}\mu_{1}^{2} +\frac{1}{2}\frac{\lambda_{1}%
(\lambda_{2}-\beta_{1})(\lambda_{2}-\beta_{2})}{\lambda_{2}-\lambda_{1}}%
\mu_{2}^{2}-(\beta_{1}+\beta_{2})\lambda_{1}\lambda_{2} +\lambda_{1}%
\lambda_{2}(\lambda_{1}+\lambda_{2}),
\]
and hence, according to (\ref{qhk}), the associated quantum operators in these
coordinates read
\begin{align*}
\widehat{H}= & \frac{1}{2}\hslash^{2}\biggl[ \frac{(\lambda_{1}-\beta
_{1})(\lambda_{1}-\beta_{2})}{\lambda_{1}-\lambda_{2}}\frac{\partial^{2}%
}{\partial\lambda_{1}^{2}} +\frac{(\lambda_{2}-\beta_{1})(\lambda_{2}%
-\beta_{2})}{\lambda_{2}-\lambda_{1}}\frac{\partial^{2}}{\partial\lambda
_{2}^{2}}\\
& +\frac{\lambda_{1}-\frac{1}{2}(\beta_{1}+\beta_{2})}{\lambda_{1}-\lambda
_{2}}\frac{\partial}{\partial\lambda_{1}}+\frac{\lambda_{2}-\frac{1}{2}%
(\beta_{1}+\beta_{2})}{\lambda_{2}-\lambda_{1}}\frac{\partial}{\partial
\lambda_{2}}\biggr]+(\beta_{1}+\beta_{2})(\lambda_{1}+\lambda_{2})
+\lambda_{1}\lambda_{2}-(\lambda_{1}+\lambda_{2})^{2},
\end{align*}
\begin{align*}
\widehat{F}= & -\frac{1}{2}\hslash^{2}\biggl[ \frac{\lambda_{2}(\lambda
_{1}-\beta_{1})(\lambda_{1}-\beta_{2})}{\lambda_{1}-\lambda_{2}}\frac
{\partial^{2}}{\partial\lambda_{1}^{2}} +\frac{\lambda_{1}(\lambda_{2}%
-\beta_{1})(\lambda_{2}-\beta_{2})}{\lambda_{2}-\lambda_{1}}\frac{\partial
^{2}}{\partial\lambda_{2}^{2}}\\
& +\frac{\lambda_{1}\lambda_{2}-\frac{1}{2}(\beta_{1}+\beta_{2})\lambda_{2}%
}{\lambda_{1}-\lambda_{2}}\frac{\partial}{\partial\lambda_{1}}+\frac
{\lambda_{1}\lambda_{2}-\frac{1}{2}(\beta_{1}+\beta_{2})\lambda_{1}}%
{\lambda_{2}-\lambda_{1}}\frac{\partial}{\partial\lambda_{2}}\biggr]-(\beta
_{1}+\beta_{2})\lambda_{1}\lambda_{2} +\lambda_{1}\lambda_{2}(\lambda
_{1}+\lambda_{2}).
\end{align*}
One can readily check that $[\widehat{H},\widehat{F}]=0$. The joint eigenvalue
problem for $\widehat{H}$ and $\widehat{F}$
\begin{align}
\label{1} & \widehat{H}\Psi(\lambda_{1},\lambda_{2})=E\Psi(\lambda_{1}%
,\lambda_{2})\\
& \widehat{F}\Psi(\lambda_{1},\lambda_{2})=\overline{E}\Psi(\lambda
_{1},\lambda_{2})\label{2}%
\end{align}
separates into two copies of the following one-dimensional eigenvalue
problem:
\[
\frac{1}{2}\hslash^{2}(\lambda-\beta_{1})(\lambda-\beta_{2})\psi
^{^{\prime\prime}}(\lambda)+\frac{1}{2}\hslash^{2}[\lambda-\frac{1}{2}%
(\beta_{1}+\beta_{2})]\psi^{^{\prime}}(\lambda) -[\lambda^{3}-(\beta_{1}%
+\beta_{2})\lambda^{2}+E\lambda+\overline{E}]\psi(\lambda)=0.
\]
Here $\lambda=\lambda_{1},\lambda_{2}$.
%; if $\lambda=\lambda_i$ then $\psi=\psi(\lambda_i)$.
The joint eigenfunction of $\widehat{H}$ and $\widehat{F}$ has the form
$\Psi(\lambda_{1},\lambda_{2})=\psi(\lambda_{1})\psi(\lambda_{2})$, and the
separation parameters in (\ref{1}) and (\ref{2}) are $\overline{E}$ and $E$ respectively.

\section{On other admissible quantizations}

As far as the quantization procedure is concerned, at the mathematical level
of the theory there are many admissible quantizations leading to different
forms of Hamiltonian operators. Apparently there is no way of telling from the
first principles which one is appropriate; this can be verified through
experiment only.

On the other hand, the number of known physical quantum systems with finitely
many degrees of freedom being counterparts of some classical systems is very
limited. These systems are mostly described by the so-called natural
Hamiltonians with flat metrics \eqref{nh}. This \emph{per se} does not suffice
to fix uniquely the quantization and thus leads to ambiguities. Thus, one
encounters in the literature various quantizations which coincide for the
class of natural Hamiltonians with flat metrics.

The choice made in the present paper for the quantum version \eqref{qhk} of a
classical constant of motion \eqref{hk}, also made in \cite{Ben3,Ben4,Du}, is
called a \emph{minimal} quantization \cite{Du,bla4}. In the literature one can
also find other quantizations of Hamiltonians \eqref{hk} (see \cite{bla4} and
references therein) of a general form
\begin{equation}
\hat{H} = -\frac{\hbar^{2}}{2} \biggl(\nabla_{i} K^{ij} \nabla_{j} + \frac
{1}{4} K^{ij}_{\phantom{ij};ij} - \frac{1}{4}\beta K^{ij} R_{ij} \biggr) +
V,\label{eq:a}%
\end{equation}
where $\beta\in\mathbb{R}$, $R_{ij}$ is the Ricci tensor and the subscript
preceded by semicolon indicates the covariant derivative in the appropriate
direction (e.g. $;k$ stands for the covariant derivative in the direction of
the vector field $\partial_{k}$). When $K^{ij}$ is just the metric tensor
$g^{ij}$, the above formula boils down to
\begin{equation}
\hat{H} = -\frac{\hbar^{2}}{2} \left( g^{ij} \nabla_{i} \nabla_{j} - \frac
{1}{4}\beta R \right)  + V,\label{eq:5.2}%
\end{equation}
where $R$ is the scalar curvature. From \eqref{eq:a} it is obvious that the
quantization procedure which is not a minimal one generates an extra potential
in the quantum Hamiltonian. This potential causes some troubles as in general
it is not expressible through appropriate separable potentials.

Below we give a few comments on this class of admissible quantizations without
going into any details.

First, one can prove that for $f_{i}(\lambda_{i})=\lambda_{i}^{k}$ we have
\begin{equation}
K_{r;ij}^{ij}=\frac{1}{4}(n+1-r)V_{r-1}^{(k-1)},
\end{equation}
where $K_{r}$ is a contravariant Killing tensor of the Benenti class
\eqref{ben}, $k\in\mathbb{Z}$, and $V_{r-1}^{(k-1)}$ is a separable potential \eqref{Vk},\eqref{Fmat}.

For the flat case $0\leq k \leq n$ (cf.\ the preceding section) and
$V_{r}^{(k)}=\delta_{r,n-k}$ the Hamiltonian operators \eqref{eq:a} coincide
with those arising from the minimal quantization \eqref{qhk} up to a constant.
On the other hand, for the non-flat case the extra terms $R$ and $K^{ij}
R_{ij}$ are complicated functions of coordinates and cannot be expressed
through appropriate separable potentials, so both quantum separability and
quantum integrability are destroyed. The only exception is the case of
constant curvature. Then $k=n+1$, we have
\begin{equation}
K_{s}^{ij}R_{ij}=-\frac{1}{4}(n+1-s)(n-1)V_{s-1}^{(n)},\qquad V_{0}\equiv1,
\end{equation}
and the choice $\beta=-\frac{1}{n-1}$ ensures cancelation of the extra terms
in \eqref{eq:a}.

\section*{Acknowledgements}

The research of AS was supported in part by the Grant Agency of the Czech
Republic (GA \v CR) under grant P201/11/0356 and by the Ministry of Education,
Youth and Sport of the Czech Republic (M\v SMT \v CR) under RVO funding for
I\v{C}47813059. AS is also very pleased to thank Maciej B\l aszak and
B\l a\.zej Szablikowski and the Faculty of Physics as a whole for the warm
hospitality extended to him in the course of his visits to the Adam Mickiewicz University.\looseness=-1

\end{document}